\documentclass{article}
\pdfoutput=1
\usepackage[utf8]{inputenc}
\usepackage{amsmath}
\usepackage{amsfonts}
\usepackage{amssymb}
\usepackage[pdflatex=true]{gastex}   
\usepackage{graphicx}
\usepackage{subfig}
\usepackage{comment}
\usepackage{amsthm}     

\newtheorem{theorem}{Theorem}

\newtheorem{lemma}[theorem]{Lemma}

\newtheorem{remark}{Remark}

\def\fract#1/#2{\leavevmode
 \kern.1em \raise .5ex \hbox{\the\scriptfont0 #1}%
 \kern-.1em $/$%
 \kern-.15em \lower .25ex \hbox{\the\scriptfont0 #2}%
}
\def\abs#1{\ensuremath{\lvert #1\rvert}} 

\makeatletter
\DeclareRobustCommand\sfrac[1]{\@ifnextchar/{\@sfrac{#1}}%
                                            {\@sfrac{#1}/}}
\def\@sfrac#1/#2{\leavevmode\scalebox{.9}{\kern.1em\raise.5ex
         \hbox{$\m@th\mbox{\fontsize\sf@size\z@
                           \selectfont#1}$}\kern-.1em
         /\kern-.15em\lower.25ex
          \hbox{$\m@th\mbox{\fontsize\sf@size\z@
                            \selectfont#2}$}}}
\DeclareRobustCommand\numfrac[1]{\@ifnextchar/{\@numfrac{#1}}%
                                            {\@numfrac{#1}}}
\def\@numfrac#1{\leavevmode \hbox{$\m@th\mbox{\fontsize\sf@size\z@
                           \selectfont#1}$}}
\makeatother
\newcommand{\nat}{\mathbb N}
\newcommand{\real}{\mathbb R}
\newcommand{\tuple}[1]{\langle #1 \rangle}

\newcommand{\true}{{\sf true}}

\renewcommand{\L}{{\cal L}}
\newcommand{\T}{{\mathcal T}}

\newcommand{\nextq}{\mathop{\mathcal X\!\!\;Q} }
\newcommand{\qnext}{\mathop{Q \mathcal X}}
\newcommand{\anext}{\mathop{\forall \mathcal X}}
\newcommand{\enext}{\mathop{\exists \mathcal X}}
\newcommand{\until}{\mathop{\mathcal U}}

\newcommand{\auntil}{\mathop{\forall \mathcal U}}
\newcommand{\euntil}{\mathop{\exists \mathcal U}}
\newcommand{\quntil}{\mathop{Q \mathcal U}}
\newcommand{\untila}{\mathop{\mathcal U\!\!\; \forall}}
\newcommand{\untile}{\mathop{\mathcal U\!\!\; \exists}}
\newcommand{\untilq}{\mathop{\mathcal U\!\!\;Q}}

\def\sg{\mathrel[\joinrel\mathrel[}
\def\sd{\mathrel]\joinrel\mathrel]}
\newcommand{\sem}[1]{\sg \mathrel{#1} \sd}

\title{{\bf Computation Tree Logic for Synchronization 
Properties}\footnote{This research was partially supported by Austrian Science Fund (FWF) 
NFN Grant No S11407-N23 (RiSE/SHiNE), ERC Start grant (279307: Graph Games), 
Vienna Science and Technology Fund (WWTF) through project ICT15-003,
and European project Cassting (FP7-601148).
}}

\author{
Krishnendu Chatterjee$^\dag$ \quad  Laurent Doyen$^{\S}$ \\ 
\normalsize
 $\strut^\dag$ IST Austria \quad $\strut^\S$ CNRS \& LSV, ENS Cachan 
}

\date{}

\begin{document}
\sloppy
\providecommand*{\donothing}[1]{}

\maketitle 
\pagestyle{plain}

\begin{abstract}
We present a logic that extends CTL (Computation Tree Logic) 
with operators that express synchronization properties.
A property is synchronized in a system if it holds in all 
paths of a certain length. 
The new logic is obtained by using the same path quantifiers 
and temporal operators as in CTL, but allowing a different order
of the quantifiers. This small syntactic variation induces 
a logic that can express non-regular properties for which known extensions
of MSO with equality of path length are undecidable.
We show that our variant of CTL is decidable and that the model-checking
problem is in $\Delta_3^{\text{P}} = \text{P}^{\text{NP}^{\text{NP}}}$, and is hard
for the class of problems solvable in polynomial time using a
parallel access to an NP oracle.
We analogously consider quantifier exchange in extensions of CTL, and 
we present operators defined using basic operators of CTL* that express
the occurrence of infinitely many synchronization points.
We show that the model-checking problem remains in $\Delta_3^{\text{P}}$.
The distinguishing power of CTL and of our new logic coincide
if the Next operator is allowed in the logics, thus the classical 
bisimulation quotient can be used for state-space reduction before model checking.
\end{abstract}

\section{Introduction}
In computer science, it is natural to view computations as a tree, where each branch 
represents an execution trace, and all possible execution traces are arranged in a tree. 
To reason about computations, the logical frameworks that express properties of 
trees have been widely studied~\cite{CGP01,Lenzi10,Thomas97}, such as CTL, CTL*, $\mu$-calculus, MSO, etc. 
These logics can express $\omega$-regular properties about trees.

A key advantage of logics is to provide concise and formal semantics, and 
a rigorous language to express properties of a system. 
For example, the logic CTL is widely used in verification tools such as NuSMV~\cite{CCGR00}, 
and hyperproperties, i.e. tree-based properties that cannot be defined over individual traces, 
are relevant in security~\cite{CFKMRS14,CS10}.

One key property that has been studied in different contexts is the property of 
synchronization, which intuitively requires that no matter how the system behaves 
it synchronizes to a common good point. 
Note that the synchronization property is inherently a tree-based property, 
and is not relevant for traces. 
Synchronization has been studied for automata~\cite{Volkov08,CMS16}, 
probabilistic models such as Markov decision processes~\cite{DMS14a,DMS14b}, 
as well as partial-information, weighted, and timed models~\cite{LLS14,KLLS15,DJLMS14}, 
and has a rich collection of results as well as open problems, e.g., \v{C}ern\'{y}'s 
conjecture about the length of synchronizing words in automata is one of the long-standing 
and well-studied problems in automata theory~\cite{Cer64,Volkov08}. 
A natural question is how can synchronization be expressed in a logical framework.

First, we show that synchronization is a property that is not $\omega$-regular. 
Hence it cannot be expressed in existing tree-based logics, such as MSO, CTL*, etc. 
A natural candidate to express synchronization in a logical framework is to consider 
MSO with quantification over path length. Unfortunately the quantification over path length in MSO
leads to a logic for which the model-checking problem is undecidable~\cite[Theorem 11.6]{Thomas90}. 
Thus an interesting question is how to express synchronization in a logical framework 
where the model-checking problem is decidable. 

\paragraph*{Contributions}
In this work we introduce an elegant logic, obtained by a natural variation of CTL. 
The logic allows to exchange the temporal and path quantifiers in classical CTL formulas. 
For example, consider the CTL formula $\forall F q$ expressing the property 
that in all paths there exists a position where $q$ holds (quantification pattern $\forall \text{paths} \cdot \exists \text{position}$). 
In our logic, the formula $F \forall q$ with quantifiers exchanged expresses that 
there exists a position $k$ such that for all paths, $q$ holds at position $k$ (quantification pattern $\exists \text{position} \cdot \forall \text{paths}$),
see \figurename~\ref{fig:eventually}.
Thus $q$ eventually holds in all paths at the same position, expressing that
the paths are eventually synchronized. 

We show that the model-checking problem is decidable for our logic,
which we show is in $\Delta_3^{\text{P}} = \text{P}^{\text{NP}^{\text{NP}}}$ (in the third 
level of the polynomial hierarchy) and is hard
for the class $\text{P}^{\text{NP}}_{\parallel}$ of problems solvable in polynomial time using a
parallel access to an NP oracle (Theorem~\ref{theo:CTL+Sync-complexity}).
The problems in $\text{P}^{\text{NP}^{\text{NP}}}$ can be solved by a 
polynomial-time algorithm that uses an oracle for a problem in $\text{NP}^{\text{NP}}$,
and the problems in $\text{NP}^{\text{NP}}$ can be solved by a non-deterministic
polynomial-time algorithm that uses an oracle for an NP-complete problem;
the problems in $\text{P}^{\text{NP}}_{\parallel}$ can be solved by a 
polynomial-time algorithm that works in two phases, where in the first phase
a list of queries is constructed, and in the second phase the queries are
answered by an NP oracle (giving a list of yes/no answers) and the algorithm
proceeds without further calling the oracle~\cite{Wagner87,Spako05}. 

We present an extension of our logic that can express the occurrence
of infinitely many synchronization points (instead of one as in eventually
sychronizing), and the absence of synchronization from some point on, 
with the same complexity status (Section~\ref{sec:extension}). 
These properties are the analogue of the classical liveness and co-liveness 
properties in the setting of synchronization. 
We show that such properties cannot be expressed in our basic logic (Section~\ref{sec:expressivity}).
In Section~\ref{sec:CTL*}, we consider the possibility to further extend 
our logic with synchronization to CTL*, and show that the exchange of quantifiers
in CTL* formulas would lead to either a counter-intuitive semantics,
or an artificial logic that would be inelegant.

We study the distinguishing power of the logics in Section~\ref{sec:distinguising},
that is the ability of the logics, given two models, to provide a formula
that holds in one model, and not in the other. The distinguishing power is
different from the expressive power of a logic, as two logics with the
same expressive power have the same distinguishing power but not vice versa.
The distinguishing power can be used for state-space reduction before running
a model-checking algorithm, in order to obtain a smaller equivalent model,
that the logic cannot distinguish from the original model, and thus for which 
the answer of the model-checking algorithm is the same.
We show that if the Next operator is allowed in the logic, then 
the distinguishing power coincides with that of CTL (two models are
indistinguishable if and only if they are bisimilar), and if 
the Next operator is not allowed, then the distinguishing power lies
between bisimulation and stuttering bisimulation, and is NP-hard to decide.
In particular, it follows that with or without the Next operator the state-space reduction 
with respect to bisimulation, which is computable in polynomial time, is sound 
for model-checking.

\section{CTL + Synchronization}

We introduce the logic CTL+Sync after presenting basic definitions 
related to Kripke structures.
A \emph{Kripke structure} is a tuple $K = \tuple{T, \Pi, \pi, R}$
where $T$ is a finite set of states, $\Pi$ is a finite set of atomic propositions,
$\pi: T \to 2^{\Pi}$ is a labeling function that maps each state~$t$ 
to the set $\pi(t)$ of propositions that are true at~$t$,
and $R \subseteq T \times T$ is a transition relation.
We denote by $R(t) = \{t' \mid (t,t') \in R\}$ the set of successors of a state~$t$ according to~$R$,
and given a set $s \subseteq T$ of states, let $R(s) = \bigcup_{t \in s} R(t)$.
A Kripke structure is \emph{deterministic} if $R(t)$ is a singleton for all 
states $t \in T$.
A \emph{path} in $K$ is an infinite sequence $\rho = t_0 t_1 \dots$ such that $(t_i,t_{i+1}) \in R$
for all $i \geq 0$. For $n \in \nat$, we denote by $\rho + n$ the suffix $t_n t_{n+1} \dots$.

\subsection{Syntax and semantics}\label{sec:ss}
In the CTL operators, a path quantifier always precedes the temporal quantifiers 
(e.g., $\euntil$ or $\auntil$). 
We obtain the logic CTL+Sync from traditional CTL by allowing to switch the order of
the temporal and path quantifiers. For example, the CTL formula $p \auntil q$
holds in a state~$t$ if for all  paths $(\forall)$ from $t$, there is a position
where $q$ holds, and such that $p$ holds in all positions before $(\until)$.
In the CTL+Sync formula $p \untila q$, the quantifiers are exchanged, and
the formula holds in $t$ if there exists a position $k$, such that for all positions $j < k$ before $(\until)$,
in all paths $(\forall)$ from $t$, we have that $q$ holds at position $k$ and 
$p$ holds at position~$j$, see \figurename~\ref{fig:UA}.
Thus the formula $p \untila q$ requires that $q$ holds synchronously 
after the same number of steps in all paths, while the formula $p \auntil q$
does not require such synchronicity across several paths.

The syntax of the formulas in CTL+Sync is as follows:
$$\varphi ::= p \mid \lnot \varphi_1 \mid \varphi_1 \lor \varphi_2 \mid 
\qnext \varphi_1 \mid
\varphi_1 \quntil \varphi_2 \mid 
\varphi_1 \untilq \varphi_2 
$$
where $p \in \Pi$ and $Q \in \{\exists, \forall\}$.
We define $\true$ and additional Boolean connectives as usual, 
and let 
\begin{itemize}
\item $\exists F \varphi \equiv \true \euntil \varphi$, and $F \exists \varphi \equiv \true \untile \varphi$, etc. 
\item $\exists G \varphi \equiv \lnot \forall F \lnot \varphi$, etc. 
\end{itemize}

Note that the Next operators $\qnext$ has only one quantifier,
and thus there is no point in switching quantifiers or defining
an operator $\nextq$.  

Given a Kripke structure $K = \tuple{T, \Pi, \pi, R}$, and a state $t \in T$, 
we define the satisfaction relation $\models$ as follows. The first cases
are standard and exist already in CTL:

\begin{itemize}
\item $K,t \models p$ if $p \in \pi(t)$.

\item $K,t \models \lnot \varphi_1$ if $K,t  \not\models \varphi_1$.

\item $K,t \models \varphi_1 \lor \varphi_2$ if $K,t \models \varphi_1$ or $K,t \models \varphi_2$.

\item $K,t \models \enext \varphi_1$ if $K,t' \models \varphi_1$ for some $t' \in R(t)$.

\item $K,t \models \anext \varphi_1$ if $K,t' \models \varphi_1$ for all $t' \in R(t)$.
\end{itemize}

\noindent The interesting new cases are built using the until operator of CTL:

\begin{itemize}

\item $K,t \models \varphi_1 \euntil \varphi_2$ if there exists a path $t_0 t_1 \dots$ in $K$ with $t_0 = t$ and 
there exists $k \geq 0$ such that:
$K,t_k \models \varphi_2$, and $K,t_j \models \varphi_1$ for all $0 \leq j < k$.\smallskip

\item $K,t \models \varphi_1 \untile \varphi_2$ if there exists $k \geq 0$ such that for all $0 \leq j < k$, 
there exists a path $t_0 t_1 \dots$ in $K$ with $t_0 = t$ such that $K,t_j \models \varphi_1$ and $K,t_k \models \varphi_2$.\smallskip

\item $K,t \models \varphi_1 \auntil \varphi_2$ if for all paths $t_0 t_1 \dots$ in $K$ with $t_0 = t$,
there exists $k \geq 0$ such that:
$K,t_k \models \varphi_2$, and $K,t_j \models \varphi_1$ for all $0 \leq j < k$.\smallskip

\item $K,t \models \varphi_1 \untila \varphi_2$ if there exists $k \geq 0$ such that for all $0 \leq j < k$ and 
for all paths $t_0 t_1 \dots$ in $K$ with $t_0 = t$, we have $K,t_j \models \varphi_1$ and $K,t_k \models \varphi_2$.\smallskip
\end{itemize}

We often write $t \models \varphi$ (or $K \models \varphi$)  when the Kripke structure $K$ (or the initial state $t$) 
is clear from the context.
\begin{figure}[!tb]
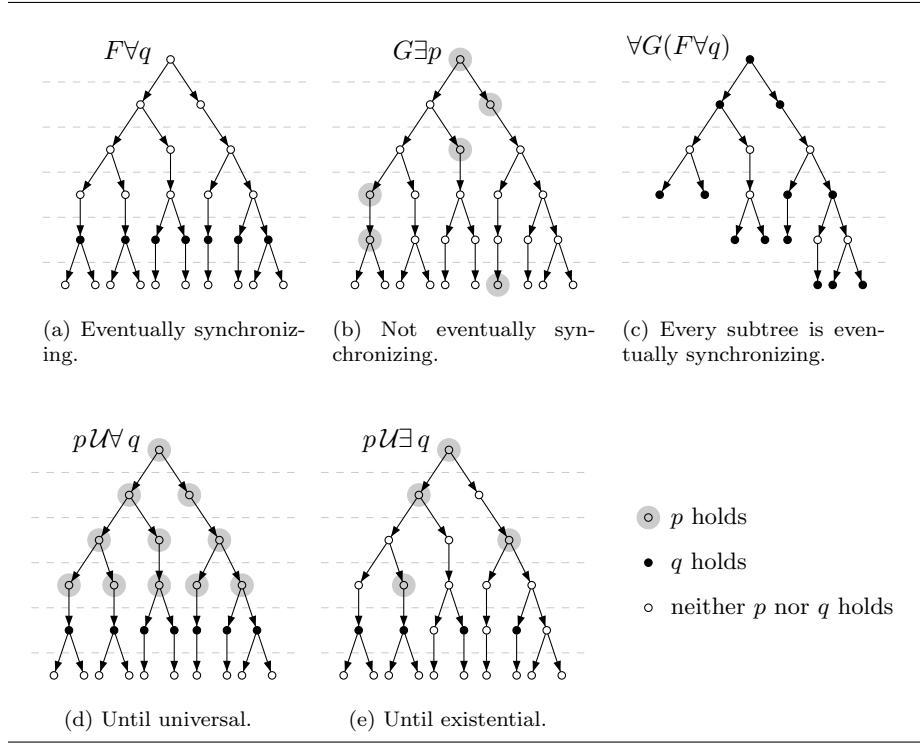
%
	\begin{center}
	\hrule
		\subfloat[Eventually synchronizing.\label{fig:eventually}]{\input{figures/exampleFA.tex}}%
		\quad
		\subfloat[Not eventually synchronizing.\label{fig:globally}]{\input{figures/exampleGE.tex}}
		\quad
		\subfloat[Every subtree is eventually synchronizing.\label{fig:subtrees}]{\input{figures/exampleAGFA.tex}}
		\quad
		\subfloat[Until universal. \label{fig:UA}]{\input{figures/exampleUA.tex}}
		\quad
		\subfloat[Until existential. \label{fig:UE}]{\input{figures/exampleUE.tex}}
		\quad
		\subfloat{\input{figures/example-caption.tex}}
	\smallskip
	\hrule
	\smallskip
		\caption{Formulas of CTL+Sync.\label{fig:example}}%
	\end{center}
\end{figure}
%
%
%
Examples of formulas are given in \figurename~\ref{fig:example}.
The examples show the first steps of the unravelling of Kripke structures
defined over atomic propositions $\{p,q\}$.
The formula $F \forall q$ expresses that $q$
eventually holds synchronously on all paths, after the same 
number of steps (\figurename~\ref{fig:eventually}). 
This is different from the CTL formula $\forall F q$, which expresses that all paths
eventually visit a state where $q$ holds, but not necessarily
after the same number of steps in all paths.
The dual formula $G \exists p$ requires that at every depth (i.e., for all positions $k$),
there exists a path where $p$ holds at depth $k$ (\figurename~\ref{fig:globally}). 
On the other hand note that $F \exists q \equiv \exists F q$ and dually $G \forall p \equiv \forall G p$.
Another example is the formula $\forall G (F \forall q)$ expressing
that every subtree is eventually synchronizing (\figurename~\ref{fig:subtrees}).
The until universal formula $p \untila q$ holds if $q$ holds at a certain position
in every path (like for the formula $F \forall q$), and $p$ holds in all
positions before (\figurename~\ref{fig:UA}). The until existential 
formula $p \untile q$ says that it is possible to find path(s)
where $q$ holds at the same position, and such that for all smaller positions 
there is one of those paths where $p$ holds at that position (\figurename~\ref{fig:UE}).

\begin{remark}\label{rmk:not-regular}
The definition of CTL+Sync, although very similar to the definition of CTL,
interestingly allows to define non-regular properties, thus not expressible in CTL (or even in MSO over trees).
It is easy to show using a pumping argument that the property $F \forall q$ 
of eventually sychronizing is not regular (\figurename~\ref{fig:eventually}).
This property of eventually sychronizing can be expressed in MSO extended with
a length predicate, by a formula such as $\exists \rho \in T^* \cdot \forall \rho' \in T^*:
\abs{\rho} = \abs{\rho'} \implies q(\rho')$ where $T= \{0,1\}$ and $q(\cdot)$ is a monadic predicate 
for the proposition $q$ over the binary tree $T^*$, where $q(\rho)$ means that $q$ holds in the last state of $\rho$.
However, model-checking for the logic MSO extended with the ``equal-length'' predicate $p$ defined by
$p(\rho, \rho') \equiv \abs{\rho} = \abs{\rho'}$ 
is undecidable~\cite[Theorem 11.6]{Thomas90}.
In contrast, we show in Theorem~\ref{theo:CTL+Sync-complexity} that the logic CTL+Sync is decidable.
\end{remark}

\subsection{Model-checking}\label{sec:mc}

Given a CTL+Sync formula $\varphi$, a Kripke structure $K$, and a state $t$, 
the \emph{model-checking problem for CTL+Sync} is to decide whether $K,t \models \varphi$ holds.

Model-checking of CTL+Sync can be decided by considering a powerset construction for
the Kripke structure, and evaluating a CTL formula on it. 
For example, to evaluate a formula $\varphi_1 \untila \varphi_2$ from state $t_I$ in a Kripke structure $K$, 
it suffices to consider the sequence $s_1 s_2 \dots$ defined by $s_1 = \{t_I\}$
and $s_{i+1} = R(s_i) $ for all $i \geq 1$, where a set $s$ is labeled by $\varphi_1$
if $K,t \models \varphi_1$ for all $t \in s$ (and analogously for $\varphi_2$).
The formula $\varphi_1 \untila \varphi_2$ holds in $t_I$ if and only if 
the formula $\varphi_1 \until \varphi_2$ holds in the sequence $s_1 s_2 \dots$
(note that on a single sequence the operators $\auntil$ and $\euntil$ are 
equivalent, thus we simply write $\until$).

For the formula $\varphi_1 \untile \varphi_2$, intuitively it 
holds in $t_I$ if there exists a set $P$ of finite paths $\rho_1, \rho_2, \dots, \rho_n$ 
from $t_I$ in $K$, all of the same length $k$, such that $\varphi_2$ holds in the last state of $\rho_i$
for all $1 \leq i \leq n$, and for every $1 \leq j < k$ there is a path $\rho_{i_j}$ such that 
$\varphi_1$ holds in the $j$th state of $\rho_{i_j}$. 
To evaluate $\varphi_1 \untile \varphi_2$ from $t_I$, we construct
the Kripke structure $2^K = \tuple{2^T, \{\varphi_1, \varphi_2\}, \pi, \hat{R}}$
where $(s,s') \in \hat{R}$ if for all $t \in s$ there exists $t' \in s'$ such that
$(t,t') \in R$, thus we have to choose (nondeterministically) at least one 
successor from each state in $s$, that is
for every set $P$ of paths $\rho_1, \rho_2, \dots, \rho_n$ as above,
there is a path $s_1, s_2, \dots, s_k$ (with $s_1 = \{t_I\}$) in $2^K$
where the sets $s_i$ are obtained by following simultaneously the 
finite paths $\rho_1, \dots, \rho_n$, thus such that $s_i$ is the set of states at position~$i$
of the paths in $P$.
The path $s_1, s_2, \dots, s_k$ in $2^K$ corresponds to a set $P$ of finite paths in $K$
that show that $\varphi_1 \untile \varphi_2$ holds if $(1)$ $\varphi_2$ holds in all states 
of $s_k$, and $(2)$~$\varphi_1$ holds in some state of $s_i$ ($i=1,\dots, k-1$).
Hence we define the labeling function $\pi$ in $2^K$ as follows: 
for all $s \in 2^T$ 
let $\varphi_2 \in \pi(s)$ if $K,t \models \varphi_2$ for all $t \in s$, and
let $\varphi_1 \in \pi(s)$ if $K,t \models \varphi_1$ for some $t \in s$.
Finally it suffices to check whether the CTL formula $\varphi_1 \euntil \varphi_2$ holds in $2^K$ from $\{t_I\}$.


This approach gives an exponential algorithm, and even a PSPACE algorithm 
by exploring the powerset construction on the fly. 
However, we show that the complexity of the model-checking problem is much below PSPACE. 
For example our model-checking algorithm for the formula $F \forall q$
relies on guessing a position $k \in \nat$ (in binary) and checking that $q$ holds on all paths at position~$k$. 
To compute the states reachable after exactly $k$ steps, we compute the $k$th
power of the transition matrix $M \in \{0,1\}^{T \times T}$ where $M(t,t') = 1$
if there is a transition from state $t$ to state $t'$. The power $M^k$ can be computed 
in polynomial time by successive squaring of $M$. For this formula, we obtain
an NP algorithm. For the whole logic, combining the guessing and squaring technique
with a dynamic programming algorithm that evaluates all subformlas, 
we obtain an algorithm in $\text{P}^{\text{NP}^{\text{NP}}}$ for the model-checking problem.
We present a hardness result for the class $\text{P}^{\text{NP}}_{\parallel}$ of
problems solvable in polynomial time using a parallel access to an NP oracle~\cite{Wagner87,Spako05}.

\begin{theorem}\label{theo:CTL+Sync-complexity}
The model-checking problem for CTL+Sync lies in $\text{P}^{\text{NP}^{\text{NP}}}$ and is $\text{P}^{\text{NP}}_{\parallel}$-hard.
\end{theorem}

\begin{proof}
{\bf Upper bound $\text{P}^{\text{NP}^{\text{NP}}}$}. 
The upper bound is obtained by a labelling algorithm, similar to a standard
algorithm for CTL model-checking, that computes the set $\sem{\varphi} = \{t \mid K,t \models \varphi\}$ 
of states that satisfy a formula $\varphi$ in a recursive manner, following the structure of the formula: 
the algorithm is executed first on the subformulas of $\varphi$, and then the
result is used to compute the states satisfying the formula $\varphi$.
For the following formulas, the labelling algorithm for CTL+Sync is the same as 
the algorithm for CTL, thus in polynomial time~\cite{Schnoebelen02}:
\begin{itemize}
\item for $\varphi = p$, we use  $\sem{\varphi} = \{t \mid p \in \Pi(t) \}$, 
\item for $\varphi = \lnot \varphi_1$, we have $\sem{\varphi}$ is the complement of $\sem{\varphi_1}$, 
\item for $\varphi = \varphi_1 \lor \varphi_2$, we have $\sem{\varphi}$ is the union of $\sem{\varphi_1}$ and $\sem{\varphi_2}$, and
\item for $\varphi = \enext \varphi_1$, we have $\sem{\varphi} = \{t \mid R(t) \cap \sem{\varphi_1} \neq \emptyset\}$,
\item for $\varphi = \anext \varphi_1$, we have $\sem{\varphi} = \{t \mid R(t) \subseteq \sem{\varphi_1}\}$,
\item for $\varphi = \varphi_1 \quntil \varphi_2$, we obtain $\sem{\varphi}$ by a reachability analysis. 
\end{itemize}

For $\varphi = \varphi_1 \untila \varphi_2$, we show that
deciding, given the sets $\sem{\varphi_1}$ and $\sem{\varphi_2}$ and a state $t$, 
whether the state $t$ belongs to $\sem{\varphi}$ is a problem in NP 
as follows: guess a position $n \leq 2^{\abs{T}}$ (represented in binary) 
and check that all paths of length~$n$ from $t$ end up in a state where $\varphi_2$ holds. 
This can be done in polynomial time by successive squaring of the transition
matrix of the Kripke structure~\cite[Theorem 6.1]{SM73}. 
Then, we need to check that $\varphi_1$ holds on all positions of all paths 
from~$t$ of length smaller than~$n$.
For $n > \abs{T}$, this condition is equivalent to ask that all reachable states 
satisfy $\varphi_1$ (because all reachable states are reachable in at most $\abs{T}$ steps), 
thus can be checked in polynomial time. 
For $n \leq \abs{T}$, we can compute in polynomial time all states reachable 
in at most $n$ steps and check that $\varphi_1$ holds in all such states. 

For $\varphi = \varphi_1 \untile  \varphi_2$, we present an algorithm in 
$\text{NP}^{\text{NP}}$ to decide if a given state $t$ belongs to $\sem{\varphi}$.
First guess a position $n \leq 2^{\abs{T}}$ (represented in binary) and 
then decide the following problem: given $n$, decide if 
for all $k < n$ there exists a path $t_0 t_1 \dots$ in $K$ such that 
$\varphi_1$ holds in $t_k$, and $\varphi_2$ holds in $t_n$. This problem
can be solved in coNP as follows: guess $k \leq n$ (represented in binary) 
and compute the set $s_k$ of states that can be reached from $t$ in exactly $k$ steps
(using matrix squaring as above). From the set $s_k \cap \sem{\varphi_1}$  
of states that satisfy $\varphi_1$ in $s_k$, compute the set of states reachable
in exactly $n-k$ steps, and check that none satisfies $\varphi_2$ to show that
the instance of the problem is negative. Thus we obtain an algorithm in
$\text{NP}^{\text{coNP}} = \text{NP}^{\text{NP}}$. 

It follows that the model-checking problem for CTL+Sync is in $\text{P}^{\text{NP}^{\text{NP}}}$.

{\bf Lower bound}.
The proof is based on the results of Lemma~\ref{lem:CTL+Sync-lower-bounds} 
that show NP-hardness of the model-checking problem for the formula $F \forall p$.
Given a Boolean propositional formulas $x_1$ in CNF, we can construct 
a Kripke structures $K_1$ such that $x_1$ is satisfiable if and only if the formula 
$F \forall p$ holds in $K_1$. 

By~\cite[Theorem~5.2]{Wagner87}, given $k$ instances $x_1, \dots, x_k$ of an NP-complete problem P (here 3SAT~\cite{Cook71}),
such that $x_i \in P$ implies $x_{i+1} \in P$ for all $1 \leq i < k$, in order to 
prove $\text{P}^{\text{NP}}_{\parallel}$-hardness\footnote{The result of~\cite[Theorem~5.2]{Wagner87} 
shows $\text{P}^{\text{NP}}_{\text{bf}}$-hardness, 
and $\text{P}^{\text{NP}}_{\text{bf}} = \text{P}^{\text{NP}}_{\parallel}$~\cite[Theorem~1]{BH91}.}  
of the model-checking problem 
it is sufficient to construct a CTL+Sync formula $\varphi$ and a Kripke structure $K$ such that 
$K \models \varphi$ if and only if $\abs{\{i \mid x_i \in P\}}$ is odd.
For each single instance $x_i$ we can construct a CTL+Sync formula $F \forall p_i$ 
and Kripke structure $K_i$ such that $K_i \models F \forall p_i$ if and only if $x_i \in P$ (i.e.,
the 3SAT formula $x_i$ is satisfiable). Since CTL+Sync is closed under Boolean
operations $\land$, $\lor$, and $\lnot$, it suffices to show that there exists
a polynomial-size Boolean formula $\psi_{\text{odd}}$ over variables $x_1, \dots, x_k$ that holds if and only if
an odd number of the variables $x_1, \dots, x_k$ are true, which is easy. 
Replacing in $\psi_{\text{odd}}$ the variables $x_i$ by $F \forall p_i$, and taking the union of 
the Kripke structures $K_1, \dots, K_k$ (merging their initial states) and labeling all 
states of $K_i$ by $\{p_1,\dots,p_{i-1},p_{i+1},\dots,p_k\}$, we obtain a Kripke structure where  
the CTL+Sync formula $\psi_{\text{odd}}$ holds if and only if the number of yes-instances
among $x_1, \dots, x_k$ is odd, thus showing $\text{P}^{\text{NP}}_{\parallel}$-hardness of 
the model-checking problem.
\end{proof}

The complexity lower bounds for the model-checking problem in Theorem~\ref{theo:CTL+Sync-complexity}
are based on Lemma~\ref{lem:CTL+Sync-lower-bounds} where we establish complexity bounds for fixed formulas.
We recall that the problems in DP can be solved by a polynomial-time algorithm 
that uses only two calls to an oracle for an NP-complete problem. 
A classical DP-complete problem is to decide, given two Boolean formulas $\psi_1$ 
and $\psi_2$, whether both $\psi_1$ is satisfiable and $\psi_2$ is valid.

\begin{lemma}\label{lem:CTL+Sync-lower-bounds}
Let $p,q \in \Pi$ be two atomic propositions. The model-checking problem is:
\begin{itemize}
\item NP-complete for the formulas $p \untila q$ and $F \forall q$, 
\item DP-hard for the formula $p \untile q$, and      
\item coNP-complete for the formula $G \exists q$.
\end{itemize}
\end{lemma}

\begin{proof}
We prove the hardness results (complexity lower bounds), since the complexity 
upper bounds follow from the proof of Theorem~\ref{theo:CTL+Sync-complexity}. 


The proof technique is analogous to the NP-hardness proof of~\cite[Theorem 6.1]{SM73},
and based on the following. Given a Boolean propositional formula $\psi$ over variables
$x_1, \dots, x_n$, consider the first $n$ prime numbers $p_1, \dots, p_n$. 
For a number $z \in \nat$, if $z\!\mod p_i \in \{0,1\}$ for all $1 \leq i \leq n$,
then the binary vector $(z\!\mod p_1, \dots, z\!\mod p_n)$ defines an assignment 
to the variables of the formula. Note that conversely, every such binary vector 
can be defined by some number $z \in \nat$ (by the Chinese remainder theorem).

{\bf NP-hardness of $F \forall q$ (and thus of $p \untila q$)}. 
The proof is by a reduction from the Boolean satisfiability problem 3SAT which is NP-complete~\cite{Cook71}.
Given a Boolean propositional formula $\psi$ in CNF, with set $C$ of (disjunctive) clauses
over variables $x_1, \dots, x_n$ (where each clause contains three 
variables), we construct a Kripke structure $K_{\psi}$ as follows:
for each clause $c \in C$, we construct a cycle $t_0, t_1, \dots, t_{r-1}$ of length
$r = p_{u} \cdot p_v \cdot p_w$ where the three variables in the clause are
$x_u$, $x_v$, and $x_w$. We call $t_0$ the origin of the cycle,
and we assign to every state $t_i$ the label~$q$ if the number $i$ 
defines an assignment that satisfies the clause $c$.
The Kripke structure $K_{\psi}$ is the disjoint union of the cycles corresponding to 
each clause, and an initial state $t_I$ with transitions from $t_I$ to the 
origin of each cycle. Note that the Kripke structure $K_{\psi}$ can be constructed in 
polynomial time, as the sum of the first $n$ prime numbers is bounded by a polynomial
in $n$: $\sum_{i=1}^{n} p_i \in O(n^2 \log n)$~\cite{BS96}.

It follows that a number $z$ defines an assignment that satisfies the formula~$\psi$ 
(i.e., satisfies all clauses of $\psi$) if and only if every path of 
length $z+1$ from $t_I$ reaches a state labelled by $q$. 
Therefore the formula $\psi$ is satisfiable if and only if $K_{\psi},t_I \models F \forall q$,
and it follows that the model-checking problem is NP-hard for the formulas 
$F \forall q$ and for $p \untila q$ (let $p$ hold in every state of $K_{\psi}$).

\begin{figure}[!tb]
  \begin{center}
    	\hrule

\begin{gpicture}(120,76)(0,0)

\gasset{Nw=1,Nh=1,Nmr=.5, ExtNL=y, NLangle=180, NLdist=1, rdist=1, loopdiam=6}   



\rpnode[Nframe=n, Nfill=y, fillgray=.95](NodeName)(15,27)(12,8){}
\rpnode[Nframe=n, Nfill=y, fillgray=.95](NodeName)(35,27)(12,8){}
\rpnode[Nframe=n, Nfill=y, fillgray=.95](NodeName)(55,27)(12,8){}
\node[Nframe=n, iangle=90, ExtNL=n, NLangle=0, NLdist=0](NodeName)(70,27){$\dots$}
\rpnode[Nframe=n, Nfill=y, fillgray=.95](NodeName)(85,27)(12,8){}

\node[Nframe=n, iangle=90, ExtNL=n, NLangle=0, NLdist=0](u10)(15,27){$c_1$}
\node[Nframe=n, iangle=90, ExtNL=n, NLangle=0, NLdist=0](u10)(35,27){$c_2$}
\node[Nframe=n, iangle=90, ExtNL=n, NLangle=0, NLdist=0](u10)(55,27){$c_3$}
\node[Nframe=n, iangle=90, ExtNL=n, NLangle=0, NLdist=0](u10)(85,27){$c_m$}

\put(100,51){\makebox(0,0)[l]{\small \begin{tabular}{l}$p$ holds \\ at $i$th state \\ of $i$th path.\end{tabular}}}
\put(100,27){\makebox(0,0)[l]{\small \begin{tabular}{l}$p$ holds \\ everywhere \\ in the cycles.\end{tabular}}}

\drawline[AHnb=0, linegray=.8, dash={1.5}0](5,67)(110,67)
\drawline[AHnb=0, linegray=.8, dash={1.5}0](5,36)(110,36)
\drawline[AHnb=0, linegray=.8, dash={1.5}0](5,18)(110,18)

\node[Nmarks=i, iangle=90, NLangle=145, Nfill=n](s0)(55,69){$p$}
\node[Nmarks=i, iangle=90, NLangle=40, Nfill=n](s0)(55,69){$t_I$}

\node[Nmarks=n, iangle=90, Nfill=n](s1)(15,64){$p$}
\node[Nmarks=n, iangle=90](s2)(35,64){}
\node[Nmarks=n, iangle=90](s3)(55,64){}
\node[Nframe=n, iangle=90, ExtNL=n, NLangle=0, NLdist=0](s4)(68,64){$\dots$}
\node[Nmarks=n, iangle=90](s4)(85,64){}

\drawedge[ELpos=50, ELside=r, curvedepth=0](s0,s1){}
\drawedge[ELpos=50, ELside=r, curvedepth=0](s0,s2){}
\drawedge[ELpos=50, ELside=r, curvedepth=0](s0,s3){}
\drawedge[ELpos=50, ELside=r, curvedepth=0](s0,s4){}

\node[Nmarks=n, iangle=90](t1)(15,59){}
\node[Nmarks=n, iangle=90, Nfill=n](t2)(35,59){$p$}
\node[Nmarks=n, iangle=90](t3)(55,59){}
\node[Nmarks=n, iangle=90](t4)(85,59){}

\drawedge[ELpos=50, ELside=r, curvedepth=0](s1,t1){}
\drawedge[ELpos=50, ELside=r, curvedepth=0](s2,t2){}
\drawedge[ELpos=50, ELside=r, curvedepth=0](s3,t3){}
\drawedge[ELpos=50, ELside=r, curvedepth=0](s4,t4){}

\node[Nmarks=n, iangle=90](s1)(15,54){}
\node[Nmarks=n, iangle=90](s2)(35,54){}
\node[Nmarks=n, iangle=90, Nfill=n](s3)(55,54){$p$}
\node[Nmarks=n, iangle=90](s4)(85,54){}

\drawedge[ELpos=50, ELside=r, curvedepth=0](t1,s1){}
\drawedge[ELpos=50, ELside=r, curvedepth=0](t2,s2){}
\drawedge[ELpos=50, ELside=r, curvedepth=0](t3,s3){}
\drawedge[ELpos=50, ELside=r, curvedepth=0](t4,s4){}

\node[Nframe=n, iangle=90](t1)(15,49){}
\node[Nframe=n, iangle=90](t2)(35,49){}
\node[Nframe=n, iangle=90](t3)(55,49){}
\node[Nframe=n, iangle=90](t4)(85,49){}

\drawedge[ELpos=50, ELside=r, curvedepth=0](s1,t1){}
\drawedge[ELpos=50, ELside=r, curvedepth=0](s2,t2){}
\drawedge[ELpos=50, ELside=r, curvedepth=0](s3,t3){}
\drawedge[ELpos=50, ELside=r, curvedepth=0](s4,t4){}

\node[Nframe=n, iangle=90, ExtNL=n, NLangle=0, NLdist=0](t1)(15,49){$\vdots$}
\node[Nframe=n, iangle=90, ExtNL=n, NLdist=0](t2)(35,49){$\vdots$}
\node[Nframe=n, iangle=90, ExtNL=n, NLdist=0](t3)(55,49){$\vdots$}
\node[Nframe=n, iangle=90, ExtNL=n, NLdist=0](t4)(85,49){$\vdots$}

\node[Nframe=n, iangle=90](s1)(15,45){}
\node[Nframe=n, iangle=90](s2)(35,45){}
\node[Nframe=n, iangle=90](s3)(55,45){}
\node[Nframe=n, iangle=90](s4)(85,45){}

\node[Nmarks=n, iangle=90](t1)(15,40){}
\node[Nmarks=n, iangle=90](t2)(35,40){}
\node[Nmarks=n, iangle=90](t3)(55,40){}
\node[Nframe=n, iangle=90, ExtNL=n, NLangle=0, NLdist=0](t4)(69,40){$\dots$}
\node[Nmarks=n, iangle=90, Nfill=n](t4)(85,40){$p$}

\drawedge[ELpos=50, ELside=r, curvedepth=0](s1,t1){}
\drawedge[ELpos=50, ELside=r, curvedepth=0](s2,t2){}
\drawedge[ELpos=50, ELside=r, curvedepth=0](s3,t3){}
\drawedge[ELpos=50, ELside=r, curvedepth=0](s4,t4){}

\node[Nmarks=n, iangle=90, Nfill=y, NLangle=270](s1)(15,35){$q$}
\node[Nmarks=n, iangle=90](s2)(35,35){}
\node[Nmarks=n, iangle=90, Nfill=y, NLangle=270](s3)(55,35){$q$}
\node[Nmarks=n, iangle=90, Nfill=y, NLangle=270](s4)(85,35){$q$}

\drawedge[ELpos=50, ELside=r, curvedepth=0](t1,s1){}
\drawedge[ELpos=50, ELside=r, curvedepth=0](t2,s2){}
\drawedge[ELpos=50, ELside=r, curvedepth=0](t3,s3){}
\drawedge[ELpos=50, ELside=r, curvedepth=0](t4,s4){}

\node[Nmarks=n, iangle=90, Nfill=y, NLangle=290](u1)(11,34){$q$}
\node[Nmarks=n, iangle=90](u2)(8,31){}
\node[Nmarks=n, iangle=90](u3)(7,27){}
\node[Nframe=n, iangle=90](u4)(8,23){}
\drawarc[dash={.2 1}0](15,27,8,220,260)

\node[Nframe=n, iangle=90](u6)(15,19){}
\node[Nmarks=n, iangle=90, Nfill=y, NLangle=120](u7)(19,20){$q$}
\node[Nmarks=n, iangle=90](u8)(22,23){}
\node[Nmarks=n, iangle=90](u9)(23,27){}
\node[Nmarks=n, iangle=90](u10)(22,31){}
\node[Nmarks=n, iangle=90](u11)(19,34){}

\drawedge[ELpos=50, ELside=r, curvedepth=0](s1,u1){}
\drawedge[ELpos=50, ELside=r, curvedepth=0](u1,u2){}
\drawedge[ELpos=50, ELside=r, curvedepth=0](u2,u3){}
\drawedge[ELpos=50, ELside=r, curvedepth=0](u3,u4){}
\drawedge[ELpos=50, ELside=r, curvedepth=0](u6,u7){}
\drawedge[ELpos=50, ELside=r, curvedepth=0](u7,u8){}
\drawedge[ELpos=50, ELside=r, curvedepth=0](u8,u9){}
\drawedge[ELpos=50, ELside=r, curvedepth=0](u9,u10){}
\drawedge[ELpos=50, ELside=r, curvedepth=0](u10,u11){}
\drawedge[ELpos=50, ELside=r, curvedepth=0](u11,s1){}

\node[Nmarks=n, iangle=90, Nfill=y, NLangle=290](u1)(31,34){$q$}
\node[Nmarks=n, iangle=90](u2)(28,31){}
\node[Nframe=n, iangle=90](u3)(27,27){}

\drawarc[dash={.2 1}0](35,27,8,190,18)

\node[Nframe=n, iangle=90](u10)(42,31){}
\node[Nmarks=n, iangle=90](u11)(39,34){}

\drawedge[ELpos=50, ELside=r, curvedepth=0](s2,u1){}
\drawedge[ELpos=50, ELside=r, curvedepth=0](u1,u2){}
\drawedge[ELpos=50, ELside=r, curvedepth=0](u2,u3){}

\drawedge[ELpos=50, ELside=r, curvedepth=0](u10,u11){}
\drawedge[ELpos=50, ELside=r, curvedepth=0](u11,s2){}

\node[Nmarks=n, iangle=90](u1)(51,34){}
\node[Nmarks=n, iangle=90, Nfill=y, NLangle=330](u2)(48,31){$q$}
\node[Nframe=n, iangle=90](u3)(47,27){}

\drawarc[dash={.2 1}0](55,27,8,190,18)

\node[Nframe=n, iangle=90](u10)(62,31){}
\node[Nmarks=n, iangle=90](u11)(59,34){}

\drawedge[ELpos=50, ELside=r, curvedepth=0](s3,u1){}
\drawedge[ELpos=50, ELside=r, curvedepth=0](u1,u2){}
\drawedge[ELpos=50, ELside=r, curvedepth=0](u2,u3){}

\drawedge[ELpos=50, ELside=r, curvedepth=0](u10,u11){}
\drawedge[ELpos=50, ELside=r, curvedepth=0](u11,s3){}

\node[Nmarks=n, iangle=90, Nfill=y, NLangle=290](u1)(81,34){$q$}
\node[Nmarks=n, iangle=90](u2)(78,31){}
\node[Nframe=n, iangle=90](u3)(77,27){}

\drawarc[dash={.2 1}0](85,28,8,190,18)

\node[Nframe=n, iangle=90](u10)(92,31){}
\node[Nmarks=n, iangle=90](u11)(89,34){}

\drawedge[ELpos=50, ELside=r, curvedepth=0](s4,u1){}
\drawedge[ELpos=50, ELside=r, curvedepth=0](u1,u2){}
\drawedge[ELpos=50, ELside=r, curvedepth=0](u2,u3){}

\drawedge[ELpos=50, ELside=r, curvedepth=0](u10,u11){}
\drawedge[ELpos=50, ELside=r, curvedepth=0](u11,s4){}

\put(7,14){\makebox(0,0)[l]{$\psi = c_1 \land c_2 \land \dots \land c_m$}} 
\put(7,9){\makebox(0,0)[l]{$c_1 = x_1 \lor x_2 \lor \lnot x_3: \text{cycle of length } r = p_1 \cdot p_2 \cdot p_3 = 2 \cdot 3 \cdot 5 = 30$}} 
\put(7,4){\makebox(0,0)[l]{$\text{satisfying assigments for } c_1: \text{0 ({\tt 000}), 1 ({\tt 111}), 10 ({\tt 010}), \dots, 25 ({\tt 110}).}$}}





\end{gpicture}

	\hrule
	\smallskip
	\caption{Reduction to show NP-hardness of $p \untile q$ in Lemma~\ref{lem:CTL+Sync-lower-bounds}.\label{fig:redcution}}
  \end{center}
\end{figure}

{\bf NP-hardness of $p \untile q$}. 
The proof is by a reduction from 3SAT~\cite{Cook71}.
The reduction is illustrated in \figurename~\ref{fig:redcution}.
Given a Boolean propositional formula $\psi$ in CNF, with set $C$ of (disjunctive) 
clauses over variables $x_1, \dots, x_n$ (where each clause contains three 
variables), we construct a Kripke structure $K$ as follows: let $m = \abs{C}$
be the number of clauses in $\psi$, and construct $m$ disjoint simple paths $\pi_i$ 
from $t_I$ of length $m+1$ (of the form $t_I, t_1, \dots, t_m$), where the last state 
of each path $\pi_i$ has a transition to the origin of a cycle corresponding to the $i$th clause (the cycles 
and their labeling are as defined in the NP-hardness proof of $F \forall q$). 
The state $t_I$ and all states of the cycles are also labelled by~$p$, 
and in the $i$th path from $t_I$, the $i$th state after $t_I$ is labelled by~$p$.
The construction can be obtained in polynomial time.

We show that $\psi$ is satisfiable if and only if the formula $p \untile q$  
holds from $t_I$ in $K$. Recall that $p \untile q$ holds
if there exists $k \geq 0$ such that for all $0 \leq j < k$, there exists
a path $t_0 t_1 \dots$ in $K$ with $t_0 = t_I$ and $K,t_j \models p$ and 
$K,t_k \models q$.

For the first direction of the proof, if $\psi$ is satisfiable,
then let $z \in \nat$ define a satisfying assignment, and let $k = m+2+z$. 
Then all paths of length $k$ from $t_I$ in $K$ end up in a state labelled by~$q$. 
Now we consider an arbitrary $j < k$ and show that there exists a path of length $k$
from $t_I$ that ends up in a state labelled by~$q$, and with the $j$th state labelled
by $p$. For $j = 0$ and for $j > m$, the conditions are satisfied by all paths,
and for $j \leq m$, the conditions are satisfied by the $j$th path from $t_I$.

For the second direction of the proof, let $k$ be a position such that for all $0 \leq j < k$, 
there exists a path $t_0 t_1 \dots$ in $K$ with $t_0 = t_I$ and 
$K,t_j \models p$ and $K,t_k \models q$. Then $k \geq m+2$ since only the states
in the cycles are labelled by~$q$. Consider the set $P$ containing, for each $j = 1,2,\dots, m$,
a path $t_I t_1 \dots$ in $K$ with $K,t_j \models p$ and $K,t_k \models q$. 
It is easy to see by the construction of $K$ that $P$ contains all the paths 
of length $k$ in $K$. Therefore, all paths of length $z = k - (m+2)$ from the
origin of each cycle end up in a state labelled by~$q$. It follows that $z$
defines an assignment that satisfies all clauses in $\psi$, thus $\psi$
is satisfiable.

{\bf DP-hardness of $p \untile q$}. 
First we present a coNP-hardness proof for $p \untile q$ that uses a reduction of the same flavor as 
in the NP-hardness of $F \forall q$,
by a reduction from the Boolean validity problem (dual of 3SAT) which is coNP-complete~\cite{Cook71}.
Given a Boolean propositional formula $\psi$ in DNF, with set $C$ of (conjunctive)
clauses over variables $x_1, \dots, x_n$ (where each clause contains 
three variables), we construct a Kripke structure $K$ with the same structure
as in the NP-hardness of $F \forall q$, 
only the labeling is different, and defined
as follows: in each cycle $t_0, t_1, \dots, t_{r-1}$, we label by $q$ the last state 
$t_{r-1}$, and we label by $p$ all states $t_i$ such that either $(a)$ the number $i$ 
defines an assignment that satisfies the clause corresponding to the cycle,
or $(b)$ the number $i$ does not define a binary assignment.
Finally, we label by $p$ the initial state~$t_I$.

It follows from this construction that the following are equivalent:
\begin{itemize}
\item there exists an assignment that falsifies the formula $\psi$ (i.e., falsifies
all clauses in $\psi$);
\item there is a length $z$ such that every path of length $z$ from $t_I$ in $K$
ends up in a state that is not labeled by $p$ (in fact the number $z-1$ 
corresponds to a binary assignment by $(b)$, that falsifies all clauses by $(a)$),
that is the formula $p \untile q$ does not hold from $t_I$.
\end{itemize}

We show that $\psi$ is valid (is satisfied by all assignments) if and only if
$p \untile q$ holds from $t_I$ in $K$.

For the first direction of the proof, if $\psi$ is valid,
then let $k = p_1 \cdot p_2 \cdots p_n$ and note that every path
of length $k$ from $t_I$ ends up in a state labeled by $q$. Now for every
$0 < j < k$, either $j-1$ defines a binary assignment and since $\psi$ is valid, 
the assignment satisfies some clause in $\psi$, or $j-1$ does not define a binary 
assignment. In both cases (and also for $j=0$), there exists a path of length $j$ 
from $t_I$ that ends up in a state labeled by $p$, that can be prolonged to 
a path of length $k$, thus ending up in a state labeled by $q$. This shows that
$p \untile q$ holds.

For the second direction of the proof, let $k$ be such that for all $0 \leq j < k$, 
there exists a path $t_0 t_1 \dots$ in $K$ with $t_0 = t_I$ and 
$K,t_j \models p$ and $K,t_k \models q$. Consider the set of all such paths
corresponding to $j = 0,1,\dots, k-1$, and the set of cycles visited by these
paths (each path starts in $t_I$ and visits one cycle). We claim that 
the subformula of $\psi$ consisting of the clauses corresponding to those cycles
is valid, which entails the validity of $\psi$. To show this, we can assume
without loss of generality that every cycle is visited for some $j$ (we can ignore
the cycles that are not visited, and remove from $\psi$ the corresponding clauses).
It follows that $k \geq p_1 \cdot p_2 \cdots p_n$, and every assignment of the 
variables in $\psi$ is represented by some number $j < k$. For all such~$j$,
there is a path of length $j-1$ that ends up in a state labeled by $p$, hence
the corresponding assignment satisfies a clause in $\psi$, thus satisfies $\psi$.
It follows that $\psi$ is satisfied by all assignments, and $\psi$ is valid. 

The DP-hardness follows by carefully  combining the NP-hardness and coNP-hardness proofs
of $p \untile q$. 

{\bf coNP-hardness of $G \exists q$}. The result follows from
the NP-hardness of $F \forall q$ since $G \exists q$ is equivalent to $\lnot F \forall \lnot q$.
\end{proof}

The complexity result of Theorem~\ref{theo:CTL+Sync-complexity} is not tight, 
with a $\text{P}^{\text{NP}^{\text{NP}}}$ upper bound and a $\text{P}^{\text{NP}}_{\parallel}$-hard lower bound.
Even for the fixed formula $p \untile q$, the gap between our $\text{NP}^{\text{NP}}$
upper bound and the DP-hardness result provides an interesting open question for future work. 

\section{Extension of CTL+Sync with Always and Eventually} \label{sec:extension}

We consider an extension of CTL+Sync with formulas of the form $\T Q \varphi$
where $\T \in \{F,G \}^+$ is a sequence of unary temporal operators Eventually (F) and Always (G).
For example, the formula $FG \forall p$ expresses strong synchronization, 
namely that from some point on, all positions on every path satisfy $p$;
the formula $GF \forall p$ expresses weak synchronization, namely that
there are infinitely many positions such that, on every path at those positions
$p$ holds. In fact only the combination of operators $FG$ and $GF$ need to be 
considered, as the other combinations of operators reduce to either 
$FG$ or $GF$ using the LTL identities $FGF \varphi \equiv GF \varphi$ 
and $GFG \varphi \equiv FG \varphi$. 

Formally, define:

\begin{itemize}
\item $K,t \models GF \forall \varphi_1$ if for all $k \geq 0$, there exists $j \geq k$ 
such that for all paths $t_0 t_1 \dots$ in $K$ with $t_0 = t$, we have $K,t_j \models \varphi_1$.

\item $K,t \models GF \exists \varphi_1$ if for all $k \geq 0$, there exists $j \geq k$ and
there exists a path $t_0 t_1 \dots$ in $K$ with $t_0 = t$ such that $K,t_j \models \varphi_1$.



\item $K,t \models FG \forall \varphi_1$ if $K,t \not\models GF \exists \lnot \varphi_1$.

\item $K,t \models FG \exists \varphi_1$ if $K,t \not\models GF \forall \lnot \varphi_1$.
\end{itemize}

The model-checking problem for the formula $GF \forall \varphi_1$ is NP-complete:
guess positions $n,k \leq 2^{\abs{T}}$ (represented in binary) and 
check in polynomial time that the states reachable by all paths of length $n$ satisfy $\varphi_1$,
and that set of the states reachable after $n+k$ steps is the same as the set
of states reachable after $n$ steps, where $k > 0$. 
This corresponds to finding a lasso in the subset construction for the Kripke 
structure~$K$. A matching NP lower bound follows from the reduction in the
NP-hardness proof of $F \forall q$ (Lemma~\ref{lem:CTL+Sync-lower-bounds}).

The model-checking problem for the formula $GF \exists \varphi_1$ can be solved
in polynomial time, as this formula is equivalent to saying that there exists
a state labeled by $\varphi_1$ that is reachable from a reachable non-trivial 
strongly connected component (SCC) --- an SCC is trivial if it consists of a single state without self-loop. 
To prove this, note that if a state $t^*$ labeled by $\varphi_1$ is reachable 
from a reachable non-trivial SCC, then $t^*$ can be reached by an arbitrarily 
long path, thus the formula $GF \exists \varphi_1$ holds. For the other direction,
if no state labeled by~$\varphi_1$ is reachable from a reachable non-trivial SCC,
then every path to a state labeled by $\varphi_1$ is acyclic (otherwise, the path 
would contain a cycle, belonging to an SCC). Since acyclic paths have length at
most $\abs{T}$, it follows that the formula $GF \exists \varphi_1$ does not
hold, which concludes the proof.

From the above arguments, it follows that the complexity status of the 
model-checking problem for this extension of CTL+Sync is the same as the 
complexity of CTL+Sync model-checking in Theorem~\ref{theo:CTL+Sync-complexity}.

\begin{theorem}\label{theo:CTL+Sync-extended-complexity}
The model-checking problem for CTL+Sync extended with sequences of unary 
temporal operators lies in $\text{P}^{\text{NP}^{\text{NP}}}$ and is $\text{P}^{\text{NP}}_{\parallel}$-hard.
\end{theorem}

\section{Expressive Power}\label{sec:expressivity}

The expressive power of CTL+Sync (even extended with Always and Eventually) 
is incomparable with the expressive power of MSO. 
By the remark at the end of Section~\ref{sec:ss}, CTL+Sync can express non-regular
properties, and thus is not subsumed by MSO, and standard argument based on counting properties~\cite{Wolper83}
showing that CTL is less expressive than MSO apply straightforwardly to show that 
formulas of MSO are not expressible in CTL+Sync~\cite{CGP01}.

We show that the formulas $GF \forall p$ and $FG \forall p$ for weak and strong synchronization
cannot be expressed in the logic CTL+Sync, thus CTL+Sync extended with sequences of unary 
temporal operators is strictly more expressive than CTL+Sync. 
The result holds if the Next operator is not allowed, and also if the 
Next operator is allowed.

\begin{figure}[!tb]%
	\begin{center}
    	\hrule
		\input{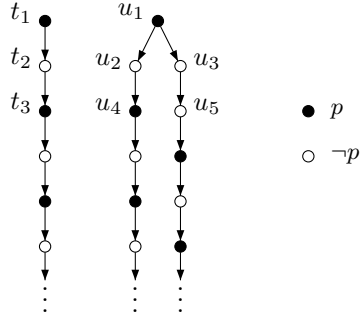}%
	\hrule
	\smallskip
	\caption{States $t_1$ and $u_1$ are indistinguishable by formulas of CTL+Sync.\label{fig:FG}}%
	\end{center}
\end{figure}

\begin{theorem}\label{theo:CTL+Sync-extended-expressive-power-2}
The logic CTL+Sync (even without the Next operator) extended with sequences of unary 
temporal operators is strictly more expressive than CTL+Sync (even using the Next operator).
\end{theorem}

\begin{proof}
We show that the formula $GF \forall p$ cannot be expressed in CTL+Sync, even using
the Next operator. To prove this, given an arbitrary CTL+Sync formula $\varphi$,
we construct two Kripke structures such that $\varphi$ holds in both Kripke structures,
but the formula $GF \forall p$ holds in one and not in the other. It follows that 
$\varphi$ is not equivalent to $GF \forall p$.

Given the formula $\varphi$, we construct the two Kripke structures as follows.
First, consider two Kripke structures whose unravelling is shown in \figurename~\ref{fig:FG}
where the states reachable from $t_1$ are satisfying alternately $p$ and $\lnot p$,
and the states reachable from $u_2$ and $u_5$ are satisfying alternately $\lnot p$ and $p$.
Call black states the states where $p$ holds, and white states the states where $\lnot p$ holds.
If $n$ is the maximum number of nested Next operators in $\varphi$, 
then we construct the $n$-stuttering of the two Kripke structures in \figurename~\ref{fig:FG},
where the $n$-stuttering of a Kripke structure $K = \tuple{T, \Pi, \pi, R}$
is the Kripke structure $K^n = \tuple{T \times \{1,\dots,n\}, \Pi, \pi^n, R^n}$ 
where $\pi^n(t,i) = \pi^n(t)$ for all $1 \leq i \leq n$, and 
the transition relation $R^n$ contains
all pairs $((t,i),(t,i+1))$ for all $t \in T$ and $1 \leq i < n$,
and all pairs $((t,n),(t',1))$ for all $(t,t') \in R$.

We claim that the formula $\varphi$ holds either in both $(t_1,1)$ and $(u_1,1)$,
or in none of $(t_1,1)$ and $(u_1,1)$, while the formula $GF \forall p$ holds in $(t_1,1)$ and not in $(u_1,1)$.
We show by induction on the nesting depth of CTL+Sync formulas $\varphi$ (that
have at most $n$ nested Next operators)
that $(t_1,i)$ and $(u_1,i)$ are equivalent for $\varphi$ (for all $1 \leq i \leq n$),
and that for all black states $t,u$, the copies $(t,1)$ and $(u,1)$ are equivalent for $\varphi$,
and analogously for all pairs of white states.

The result holds trivially for formulas of nesting depth $0$, that is atomic propositions.
For the induction step, assume the claim holds for formulas of nesting depth $k$,
and consider a formula $\varphi$ of nesting depth $k+1$. If the outermost operator
of $\varphi$ is a Boolean operator, or a CTL operator ($\qnext$ or $\quntil$), 
then the result follows from the induction hypothesis and the result of~\cite[Theorem~2]{KS05}
showing two paths that differ only in the number of consecutive repetitions of a state,
as long as the number of repetitions is at least $n$, are equivalent for the formulas 
with at most $n$ nested Next operators.
If the outermost operator of $\varphi$ is either $\untile$ or $\untila$, that is
$\varphi \equiv \varphi_1 \untile \varphi_2$ or $\varphi \equiv \varphi_1 \untila \varphi_2$,
then consider a state where $\varphi$ holds: either $\varphi_2$ holds in that state,
and by the induction hypothesis, $\varphi_2$ also holds in the corresponding state (that we
claimed to be equivalent),
or $\varphi_2$ holds in the states of the other color than the current state, 
and $\varphi_1$ holds on the path(s) at all positions before. By the induction 
hypothesis, at the same distance from the claimed equivalent states, 
we can find a state where $\varphi_2$ holds in all paths, and $\varphi_1$ holds on all positions before,
which concludes the proof for the induction step.
\end{proof}

\section{Distinguishing Power}\label{sec:distinguising}

Two states of a Kripke structure can be distinguished by a logic if there 
exists a formula in the logic that holds in one state but not in the other.
Each logic induces an indistinguishability relation (which is an equivalence)
on Kripke structures that characterizes the distinguishing power of the logic.
Two states $t$ and $t'$ of a Kripke structure $K$ are indistinguishable by 
a logic $\L$ if they satisfy the same formulas of $\L$, that is
$\{\varphi \in \L \mid K,t \models \varphi\} = \{\varphi \in \L \mid K,t' \models \varphi\}$.

For CTL (with the Next operator), the distinguishing power is standard bisimulation,
and for CTL without the Next operator, the distinguishing power is stuttering
bisimulation~\cite{BCG88}. Stuttering bisimulation is a variant of bisimulation
where intuitively several transitions can be used to simulate a single transition,
as long as the intermediate states of the transitions are all equivalent (for 
stuttering bisimulation). We omit the definition of bisimulation and
stuttering bisimulation~\cite{BCG88}, and in this paper we consider that they 
are defined as the distinguishing power of respectively CTL and CTL without 
the Next operator.

It is easy to show by induction on the nesting depth of formulas that the distinguishing 
power of CTL+Sync is the same as for CTL, since $(i)$ CTL+Sync contains CTL, and 
$(ii)$ if two states $t$ and $t'$
are bisimilar, there is a correspondence between the paths starting
from $t$ and the paths starting from $t'$ (for every path from $t$, there is a
path from $t'$ such that their states at position~$i$ are bisimilar, for all $i \in \nat$, and analogously for every path
of $t'$~\cite[Lemma 3.1]{BCG88}), which implies the satisfaction of the same formulas in CTL+Sync.
The same argument holds for CTL+Sync extended with unary temporal operators (Section~\ref{sec:extension}).

\begin{theorem}\label{theo:CTL+Sync-bisimulation}
Two states $t$ and $t'$ of a Kripke structure $K$ 
are indistinguishable by CTL+Sync formulas (even extended with unary temporal operators)
if and only if $t$ and $t'$ are bisimilar. 
\end{theorem}

Without the Next operator, the logic CTL+Sync has a distinguishing power that
lies strictly between bisimulation and stuttering bisimulation, as shown
by the examples in \figurename~\ref{fig:stut-bisimilar-distinguishable} 
and \figurename~\ref{fig:indistinguishable-not-bisimilar}.
Indistinguishability by CTL+Sync formulas without the Next operator implies
indistinguishability by standard CTL without the Next operator, and thus
stuttering bisimilarity. We obtain the following result.

\begin{theorem}\label{theo:CTL+Sync-without-next-indistinguishability}
The following implications hold for all states $t,t'$ of a Kripke structure $K$:
\begin{itemize}

\item if $t$ and $t'$ are bisimilar, then $t$ and $t'$ are indistinguishable 
by CTL+Sync formulas without the Next operator (even extended with unary temporal operators);

\item if $t$ and $t'$ are indistinguishable by CTL+Sync formulas without the Next operator
(even extended with unary temporal operators), then $t$ and $t'$ are stuttering bisimilar.

\end{itemize}
\end{theorem}

It follows from the first part of Theorem~\ref{theo:CTL+Sync-without-next-indistinguishability}
that the state-space reduction techniques based on computing a bisimulation quotient before evaluating
a CTL formula will work for CTL+Sync. Although the exact indistinguishability
relation for CTL+Sync is coarser than bisimulation, we show that deciding this 
relation is NP-hard, and thus it may not be relevant to compute it for
quotienting before model-checking, but rather use the polynomial-time
computable bisimulation.

\begin{figure}[!tb]
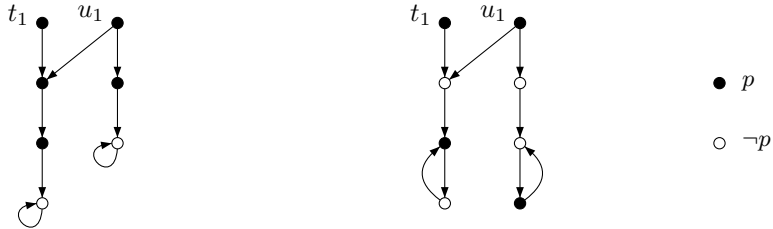
%
	\begin{center}
    \hrule
		\subfloat[The states $t_1$ and $u_1$ are stuttering bisimilar (they satisfy the same CTL formulas
without the Next operator), but they can be distinguished by the CTL+Sync formula
$p \untila \lnot p$ which holds in $t_1$ but not in $u_1$. \label{fig:stut-bisimilar-distinguishable}]{\input{figures/stut-bisimilar-distinguishable.tex}}
		\quad
		\subfloat[The states $t_1$ and $u_1$ are indistinguishable by CTL+Sync formulas,
but they are not bisimilar, i.e. they can be distinguished by CTL formulas 
with the Next operator, for example $\anext \anext p$ which holds
in $t_1$ but not in $u_1$. \label{fig:indistinguishable-not-bisimilar}]{\input{figures/indistinguishable-not-bisimilar.tex}}%
		\quad
		\subfloat{\input{figures/indistinguishable-caption.tex}}
	\hrule
	\smallskip
	\caption{The distinguishing power of CTL+Sync lies strictly between bisimulation and stuttering bisimulation. \label{fig:distinguishing}}%
	\end{center}
\end{figure}

\begin{theorem}\label{theo:CTL+Sync-without-next-NP-hard}
Deciding whether two states of a Kripke structure are indistinguishable 
by CTL+Sync formulas without the Next operator is NP-hard.
\end{theorem}

\begin{proof}
The proof is by a reduction from the Boolean satisfiability problem 3SAT which is NP-complete~\cite{Cook71}.
Given a Boolean propositional formula $\psi$ in CNF, we construct 
two Kripke structures $K$ and $K_{\psi}$ that are indistinguishable (from their initial state)
if and only if $\psi$ is satisfiable, where:
\begin{itemize}
\item $K$ is the Kripke structure shown in \figurename~\ref{fig:indistinguishable}, and 
\item $K_{\psi}$ is the Kripke structure constructed in the NP-hardness proof of $F \forall q$ (Lemma~\ref{lem:CTL+Sync-lower-bounds}).
\end{itemize}

We assume that $\psi$ contains at 
least one clause with both a positive and a negative literal 
and one clause with only negative literals. This assumption 
induces no loss of generality because we can always add two
such clauses to a formula without changing its satisfiability 
status (e.g., by introducing new variables).
This assumption ensures that at least one successor of $t_I$ satisfies~$q$ 
and at least one successor of $t_I$ satisfies $\lnot q$, exactly like from $u_I$ in $K$. 
This ensures, for example, that $t_I$ and $u_I$ agree on the formula $\lnot q \untila q$ 
(which does not hold). 

We know from the proof of Lemma~\ref{lem:CTL+Sync-lower-bounds}
that $K_{\psi}, t_I \models F \forall q$ if and only if $\psi$ is satisfiable.
Hence, it suffices to show that $K$ and $K_{\psi}$ are indistinguishable
if and only if the formula $F \forall q$ holds in $t_I$. 
Since the formula $F \forall q$ holds in $u_I$, we only need to show that 
if $F \forall q$ holds in $t_I$, then $K$ and $K_{\psi}$ are indistinguishable.

To do this, we assume that $F \forall q$ holds in $t_I$, and we show that 
for all CTL+Sync formulas $\varphi$ without the Next operator, 
$t_I \models \varphi$ if and only if $u_I \models \varphi$.

First note that from every state $t$ of $K$ or $K_{\psi}$ that is not
an initial state ($t \neq t_I$ and $t \neq u_I$), the transitions are
deterministic (i.e., there is only one path starting from $t$), and it follows
that the existential path quantifiers is equivalent to the universal path quantifier
and they can be freely switched with other quantifiers
and put in front of formulas without changing their truth value. 
Hence the operators $\auntil$ and $\untila$ are equivalent, 
as well as $\euntil$ and $\untile$, and thus CTL+Sync is the same as just CTL in those states. 
Moreover, the sequence of states on those paths alternate (up to stuttering) 
between states that satisfy $q$ and states that satisfy $\lnot q$ (or vice versa).
Therefore, all such states that satisfy~$q$ are stuttering bisimilar and thus
satisfy all the same CTL+Sync formulas, and all such states that satisfy $\lnot q$ are stuttering 
bisimilar and thus also satisfy all the same CTL+Sync formulas~\cite{BCG88}.

\begin{figure}[!tb]
  \begin{center}
	\hrule
		\input{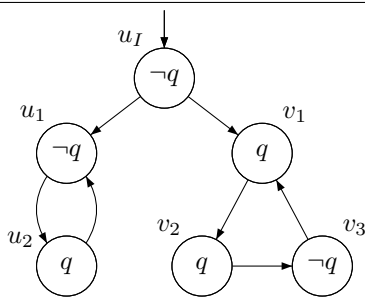}
	\hrule
	\smallskip
	\caption{The Kripke structure $K$ in the proof of Theorem~\ref{theo:CTL+Sync-without-next-NP-hard}. \label{fig:indistinguishable}}
  \end{center}
\end{figure}

Now we prove by induction on the structure of the CTL+Sync formulas $\varphi$ that 
the states $t_I$ and $u_I$ satisfy the same formulas, and thus $K$ and $K_{\psi}$ 
are indistinguishable. 
The result holds when $\varphi$ is an atomic proposition. 
We assume that the result holds for all CTL+Sync formulas
with at most $k$ nested operators, and we consider an arbitrary formula $\varphi$
with at most $k+1$ nested operators, and show that the result hold for $\varphi$.

Note that by the induction hypothesis, we can partition the states into three kinds 
such that all states of a kind satisfy the same formulas with at most $k$ nested operators:
the states of the first kind are the initial states $t_I$ and $u_I$,
the second kind is the states that are not initial and satisfy $\lnot q$,
and the third kind is the states that satisfy $q$ (which are also not initial).
Hence the formulas with at most $k$ nested operators can be classified
according to which kind of states satisfy them, thus in $2^3 = 8$ classes.
Therefore, the proof for the induction case boils down to a tedious but straightforward
case analysis, summarized as follows: 


\begin{itemize}
\item If $\varphi \equiv \lnot \varphi_1$ or $\varphi \equiv \varphi_1 \lor \varphi_2$,
then the claim holds easily by the induction hypothesis.

\item If $\varphi \equiv \varphi_1 \quntil \varphi_2$, then 

\begin{itemize}
\item either $\varphi_2$ holds in the states that satisfy $q$, and for example 
$\varphi_1$ holds in the states that satisfy $\lnot q$. Then $\varphi$
holds in all states, thus both in $t_I$ and in $u_I$;

\item or $\varphi_2$ holds in the states that satisfy $\lnot q$, and then $\varphi_2$
(and also $\varphi$) holds in all states that satisfy $\lnot q$, by the induction
hypothesis. Hence $\varphi$ holds both in $t_I$ and in $u_I$.
\end{itemize}

\item If $\varphi \equiv \varphi_1 \untila \varphi_2$, then the interesting case 
is when $\varphi_2$ holds in the states that satisfy $q$, and for example 
$\varphi_1$ holds in the states that satisfy $\lnot q$. By our assumption
on the Boolean CNF formula $\psi$, we have that $\varphi$ does not hold in 
$t_I$ nor in $u_I$. Other cases are treated analogously, in particular
for $\varphi_1 \equiv \true$, we have $\varphi \equiv F \forall q$, 
which holds in $t_I$ and in $u_I$.

\item If $\varphi \equiv \varphi_1 \untile \varphi_2$, the analysis is analogous to the previous case.

\end{itemize}

%
%
\end{proof}

\section{CTL* + Synchronization}\label{sec:CTL*}
CTL* is a branching-time extension of LTL (and of CTL) where several nested temporal 
operators and Boolean connectives can be used under the scope of a single
path quantifier. For example the CTL* formula $\exists(G \varphi \to G \psi)$ says that 
there exists a path in which either $\varphi$ does not hold in every position,
or $\psi$ holds at every position. Note that $\varphi$ and $\psi$ may also
contain path quantifiers.

Extending CTL+Sync with formula quantification analogous to CTL* 
presents some difficulties. Even considering only Boolean connectives and $\{F,G\}$ operators
leads to a logic that is hard to define. For example, one may consider
a formula like $(F p \lor F q) \forall$ which could be 
naturally interpreted as there exist two positions $m,n \geq 0$ such that 
on all paths $\rho$, either $p$ holds at position $m$ in $\rho$, or $q$ holds at
position $n$ in $\rho$. In this definition the $\lor$ operator would not be idempotent, 
that is $\psi_1 = (F p \lor F p) \forall$ is not equivalent to $\psi_2 = (F p) \forall$,
where $\psi_1$ means that the set of all paths can be partitioned into
two sets of paths where $p$ holds synchronously at some position, 
but not necessarily the same position in both sets, while $\psi_2$ expresses the
property that $p$ holds synchronously at some position in all paths. 

Another difficulty with binary operators is the semantics induced by the order of the operands. 
For instance, the formula $(F p \lor G q) \forall$
can be interpreted as $(i)$ there exists a position $m \geq 0$ such that for all positions $n \geq 0$,
on all paths $\rho$, either $\rho + m \models  p$ or $\rho + n \models q$;
or it can be interpreted as $(ii)$ for all $n \geq 0$, there exists $m \geq 0$ 
such that on all paths $\rho$, either $\rho + m \models  p$ or $\rho + n \models q$.
These two interpretations differ on the Kripke structure that produces exactly
two paths $\rho_1$ and $\rho_2$ such that $p$ and $q$ hold at the following 
positions ($p$ holds nowhere except at position $1$ in $\rho_1$ and position $3$
in $\rho_2$, and $q$ holds everywhere except position $2$ in $\rho_1$ and position $4$
in $\rho_2$): 
\begin{center}
\begin{tabular}{lccccccc}
in $\rho_1$: & $\{\bar p, q \}$ & $\{p, q \}$ & $\{\bar p, \bar q\}$ &  $\{\bar p, q \}$ &  $\{\bar p, q \}$ & $\{\bar p, q \}$ & \dots \\
in $\rho_2$: & $\{\bar p, q \}$ & $\{\bar p, q \}$ & $\{\bar p, q\}$ &  $\{p, q \}$ &  $\{\bar p, \bar q \}$ & $\{\bar p, q \}$ & \dots \\
             & 0 & 1 & 2 & 3 & 4 & 5 &  \\
\end{tabular}
\end{center} 
Note that the two paths agree on their initial position, and we can
construct a Kripke structure that produces exactly those two paths.
It is easy to see that the formula $(F p \lor G q) \forall$ does not hold 
according to the first interpretation (indeed, for $m=1$ we can take $n=4$ and consider the path $\rho_2$
where $p$ does not hold at position $1$ and $q$ does not hold at position $4$,
and for all other values of $m$, take $n=2$ and consider the path $\rho_1$
where $p$ does not hold at position $m$ and $q$ does not hold at position $2$), 
but it does hold according to the second interpretation
(for $n=2$ take $m=1$, for $n=4$ take $m=3$, and for all other values
of $n$ take arbitrary value of $m$, for example $m=n$).
The trouble is that the order of the existential
quantifier (associated to the left operand $F p$) and the universal quantifier
(associated to the right operand $G q$) actually matters in the semantics
of the formula, leading to an annoying situation that $(F p \lor G q) \forall$
is not equivalent to $(G q \lor F p) \forall$ in any of the interpretations.
One way could be to use the branching Henkin quantifiers, like ${\exists m} \choose {\forall n}$
where the existential choice of $m$ does not depend on the universal choice of $m$.
This interpretation suffers from lack of symmetry, as the negation of such 
a branching Henkin quantifier is in general not expressible as a 
branching Henkin quantifier~\cite{BG86}.


\section{Conclusion}

The logic CTL+Sync and its extensions presented in this paper provide an elegant 
framework to express non-regular properties of synchronization. 
It is intriguing that the exact optimal complexity of the model-checking problem remains open,
specially even for the fixed formula $p \untile q$ (which we show is in $\text{NP}^{\text{NP}}$,
and DP-hard).
Extending CTL+Sync to an elegant logic \emph{\`{a} la} CTL* seems challenging.
One may want to express natural properties with the flavor of synchronization, such as the 
existence of a fixed number of synchronization points, or the property that all 
paths synchronize in either of a finite set of positions, etc. (see also Section~\ref{sec:CTL*}).
Another direction is to consider alternating-time temporal logics (ATL~\cite{AHK02})
with synchronization. ATL is a game-based extension of CTL for which the 
model-checking problem remains in polynomial time. For instance, ATL can express the existence
of a winning strategy in a two-player reachability game. For the synchronized version of 
reachability games (where the objective for a player is to reach a target
state after a number of steps that can be fixed by this player, independently of the
strategy of the other player), it is known that deciding the winner is PSPACE-complete~\cite{DMS14a}.
Studying general game-based logics such as ATL or strategy logic~\cite{CHP10} 
combined with quantifier exchange is an interesting direction for future work.\medskip


{\bf Acknowledgment.} We thank Stefan G\"oller and anonymous reviewers for 
their insightful comments and suggestions.


\bibliographystyle{plain}
\bibliography{biblio} 








\end{document}